\documentclass[12pt]{article}

\usepackage{amssymb,amsmath,amsfonts,amssymb}
\usepackage{graphics,graphicx,color,hhline}
\usepackage{bbm} 
\usepackage{euscript}  
\usepackage{authblk}
 \usepackage{mathtools}
\usepackage{cases}
\usepackage{hyperref}
\def\@abssec#1{\vspace{.05in}\footnotesize \parindent .2in 
{\bf #1. }\ignorespaces} 

\graphicspath{{/EPSF/}{Figures/}}   

\newtheorem{theorem}{Theorem}[section]
\newtheorem{lemma}[theorem]{Lemma}
\newtheorem{proposition}[theorem]{Proposition}

\newtheorem{remark}[theorem]{Remark}

\def \Rm {\mathbb R}
\def \Nm {\mathbb N}

\def \NN {\mathbb N}

\newcommand{\eps}{\varepsilon}

\newcommand{\be}{\begin{equation}}
\newcommand{\ee}{\end{equation}}
\newcommand{\bea}{\begin{eqnarray}}
\newcommand{\eea}{\end{eqnarray}}
\newcommand{\bee}{\begin{eqnarray*}}
\newcommand{\eee}{\end{eqnarray*}}

\def\fref#1{{\rm (\ref{#1})}}
\newcommand{\calQ}{\mathcal Q}

\newcommand{\calL}{\mathcal L}

\newcommand{\calE}{\mathcal E}

\newcommand{\calJ}{\mathcal J}
\newcommand{\calK}{\mathcal K}
\newcommand{\calA}{\mathcal A}

\newcommand{\frakH}{\mathfrak{H}}

\newcommand{\cout}[1]{}

\DeclarePairedDelimiter\bra{\langle}{\rvert}
\DeclarePairedDelimiter\ket{\lvert}{\rangle}

\newcommand{\vs}{\varrho_\star}

\DeclareMathOperator{\Tr}{Tr}

\begin{document}
{\title{A constrained optimization problem in quantum statistical physics}}

 \author{Romain  Duboscq \footnote{Romain.Duboscq@math.univ-tlse.fr}}
 \affil{Institut de Math\'ematiques de Toulouse ; UMR5219\\Universit\'e de Toulouse ; CNRS\\INSA, F-31077 Toulouse, France}
 \author{Olivier Pinaud \footnote{pinaud@math.colostate.edu}}
 \affil{Department of Mathematics, Colorado State University\\ Fort Collins CO, 80523}

\maketitle

\begin{abstract}
In this paper, we consider the problem of minimizing quantum free energies under the constraint that the density of particles is fixed at each point of $\Rm^d$, for any $d\geq 1$. We are more particularly interested in the characterization of the minimizer, which is a self-adjoint nonnegative trace class operator, and will show that it is solution to a nonlinear self-consistent problem. This question of deriving quantum statistical equilibria is at the heart of the quantum hydrodynamical models introduced by Degond and Ringhofer in \cite{DR}. An original feature of the problem is the local nature of constraint, i.e. it depends on position, while more classical models consider the total number of particles in the system to be fixed. This raises difficulties in the derivation of the Euler-Lagrange equations and in the characterization of the minimizer, which are tackled in part by a careful parametrization of the feasible set.

\end{abstract}

\section{Introduction}
This work is concerned with the minimization of quantum free energies of the form
\be
\label{freeE}
F(\varrho)=\Tr(H \varrho)+ T \Tr(\beta(\varrho)),
\ee
where $\varrho$ is a density operator, i.e. a self-adjoint, trace class, and nonnegative operator on some Hilbert space, $H$ is a given Hamiltonian, $T$ the temperature, $\Tr(\cdot )$ denotes the operator trace, and $\beta$ is an entropy function, for instance the Boltzmann or the Fermi-Dirac entropy. The free energy $F(\varrho)$ is minimized under a constraint of local density, namely the density of particles is prescribed at each point of space: if $\rho(x,y)$ is the integral kernel associated with the operator $\varrho$, then the local density, defined as $\rho(x,x)$, is fixed and equal to a given function. 

The problem considered here is the building block of the quantum hydrodynamical models introduced by Degond et al in \cite{DR}. Their strategy consists in adapting to the quantum setting the moments closure method by entropy minimization that was developed by Levermore in the context of kinetic equations \cite{levermore}. This requires the construction of quantum statistical equilibra, which are obtained by minimizing $F(\varrho)$ under appropriate constraints. We focus in this work on the local density constraint (i.e. the zero order moment of $\varrho$) explicited above, which leads to the so-called quantum drift-diffusion model, see \cite{QET}. Different models can be obtained by considering additional constraints, in particular the local current and energy constraints (first and second order moments), which lead to the quantum Euler or quantum Navier-Stokes equations. We refer to \cite{QHD-review, QHD-CMS} for more details. See also e.g. \cite{jungelbook,jungel-matthes-milisic,jungel-matthes}
for additional references on quantum hydrodynamics.

At the mathematical level, it is proved in \cite{MP-KRM}, for $H=-\Delta+V$ defined on $L^2(\Rm^d)$, with $d \geq 1$ and $V$ a given potential, that $F(\varrho)$, with $\beta$ the Boltzmann or the Fermi-Dirac entropy, admits a unique minimizer under the constraint $\rho(x,x)=n(x)$, where $n$ is nonnegative and verifies
\be \label{hypn}
\sqrt{n} \in H^1(\Rm^d), \qquad \int_{\Rm^d} n(x)dx=1, \qquad n \log n \in L^1(\Rm^d).
\ee
The first condition above is necessary for the energy to be finite (i.e. the first term in the definition of $F(\varrho)$). The second condition is not crucial and can be modified. The proof is based on compactness and convexity methods. An important ingredient is a logarithmic Sobolev inequality for systems that yields a bound from below for the free energy. This requires the third condition in \eqref{hypn}, which prevents leakage of particles at the infinity. Without this condition, the free energy is not bounded below and the minimization problem does not admit a solution. The reference \cite{MP-KRM} addresses in addition a local current constraint, while (local) density, current and energy constraints are considered in \cite{DP-JMPA} in a one-dimensional setting in a bounded (periodic) domain. The energy constraint is difficult to handle in that there is no sufficient compactness on the  minimizing sequences to directly pass to the limit in the constraint, and one has to resort to subtle monotonicity arguments inspired by thermodynamics to conclude.

Knowing from \cite{MP-KRM} that a minimizer exists and is unique, we are interested in this work in its characterization. This is actually a quite more difficult problem than just establishing well-posedness. Formal calculations, performed e.g. in \cite{QHD-review} in the case of the Boltzmann entropy and when $V=0$, yield that the minimizer $\vs$ satisfies the following self-consistent relation 
\be \label{quantM}
\vs=\exp(-H[\vs]/T),
\ee
where $H[\vs]$ is an Hamiltonian of the form
$$
H[\vs]=-\Delta+A[\vs], \qquad A[\vs]=\frac{\Delta n}{2 n}-\frac{T n[\vs \log (\vs)]+k[\vs]}{n}.
$$
Above, $n[\vs \log \vs]$ and $k[\vs]$ are respectively the local entropy and local kinetic energy, defined by
\be \label{local}
n[\vs \log (\vs)](x)=\sum_{p \in \NN} \rho_p \log (\rho_p) |\phi_p(x)|^2, \qquad k[\vs](x)=\sum_{p \in \NN} \rho_p |\nabla \phi_p(x)|^2,
\ee
where $\{\rho_p\}_{p \in \NN}$ and $\{\phi_p\}_{p \in \NN}$ are the eigenvalues and the eigenvectors of $\vs$. Throughout the paper, the eigenvalues are counted with multiplicity and form a nonincreasing nonnegative sequence that accumulates at zero. At this stage, the kernel of $\vs$ could be zero, finite, or infinite, but we will prove that it is actually zero. The solution $\vs$ obtained in \eqref{quantM} is referred to in \cite{QHD-review} as the ``quantum Maxwellian'', and $A[\vs]$ is the chemical potential. In \cite{MP-JSP}, in a periodic one-dimensional domain $\Omega$, it is proved under the assumptions that $n$ is uniformly bounded from below, i.e. $n(x) \geq \underline{n}>0$ a.e., and that $n \in H^1(\Omega)$, that the Hamiltonian $H[\vs]$ is self-adjoint in the sense of quadratic forms. This is possible since the hypotheses \eqref{hypn} on $n$ eventually lead to  $A[\vs] \in H^{-1}(\Omega)$, which, in one dimension only in general, allows one to define $H[\vs]$ in the sense of quadratic forms. In the case where the spatial domain is $\Rm^d$, the condition $n(x)\geq \underline{n}>0$ is not compatible with $n \in L^1(\Rm^d)$, and even if we had $A[\vs] \in H^{-1}(\Rm^d)$, this is in general too low a regularity to construct a self-adjoint operator using classical results such as the KLMN theorem \cite{RS-80-2} for instance. One of the main difficulties is therefore to give a proper meaning for \fref{quantM} for the low regularity self-consistent potential $A[\vs]$. One could consider adding regularity conditions on $\Delta n$ for instance, but on the one hand it is unclear how this improves the regularity of the  self-consistent terms $n[\vs \log \vs]$ and $k[\vs]$, and on the other any additional assumptions to \fref{hypn} are not natural since they are not necessary for the existence theory.

The main result of this work is to rigorously define \fref{quantM} under the minimal assumptions \fref{hypn}. We will characterize $\{\rho_p\}_{p \in \NN}$ and $\{\phi_p\}_{p \in \NN}$ and show they are obtained by minimizing an appropriate quadratic form whose closure is $- \log (\vs)$. The method of proof is based on a proper rewriting of the chemical potential $A[\vs]$ and on exploiting the obtained particular form. While a potential with a regularity as low as that of $A[\vs]$ would not in general lead to a self-adjoint operator, it is the distinct structure of $A[\vs]$ inherited from the minimization problem that allows us to justify \fref{quantM}.

The article is structured as follows: we present our main result in Section \ref{secmain}; the proof is broken down into several parts in Section \ref{proofmainth}, and the proofs of some technical lemmas are given in Section \ref{lemproof}.

\paragraph{Acknowledgments.} OP's work is supported by NSF CAREER Grant DMS-1452349. 

\section{Main result} \label{secmain}

We start by introducing some notation. 
\paragraph{Notation.} We write $(\cdot,\cdot)$ for the inner product on $L^2(\Rm^d)$, with the convention $(f,g)=\int_{\Rm^d} \overline{f} g dx$, and $\|\cdot\|$ for the corresponding norm. The free Hamiltonian $-\Delta$ is denoted by $H_0$, equipped with the domain $H^2(\Rm^d)$. $\calL(L^2(\Rm^d))$ is the space of bounded operators on $L^2(\Rm^d)$, $\calJ_1 \equiv \calJ_1(L^2(\Rm^d))$ is the space of trace class operators and $\calJ_2$ the space of Hilbert-Schmidt operators, both on $L^2(\Rm^d)$. In the sequel, we will refer to a density operator as a self-adjoint, trace class, nonnegative operator on $L^2(\Rm^d)$. For $|\varrho|=\sqrt{\varrho^* \varrho}$, we introduce the following space:
$$\calE=\left\{\varrho\in \calJ_1:\, \overline{\sqrt{H_0}|\varrho|\sqrt{H_0}}\in \calJ_1\right\},$$
where $\overline{\sqrt{H_0}|\varrho|\sqrt{H_0}}$ denotes the extension of the operator $\sqrt{H_0}|\varrho|\sqrt{H_0}$ to $L^2(\Rm^d)$. We will drop the extension sign in the sequel to ease notation. The space $\calE$ is a Banach space when endowed with the norm
$$\|\varrho\|_{\calE}=\Tr \big(|\varrho| \big)+\Tr\big(\sqrt{H_0}|\varrho|\sqrt{H_0}\big),$$
where $\Tr(\cdot)$ denotes the operator trace. Finally, the energy space is the following closed convex subspace of $\calE$:
$$\calE_+=\left\{\varrho\in \calE:\, \varrho\geq 0\right\}.$$
The local kinetic energy of $\varrho \in \calE_+$ is defined by
$$
k[\varrho](x)=\sum_{p \in \NN} \lambda_p |\nabla \psi_p(x)|^2, \qquad \textrm{with} \qquad \|k[\varrho]\|_{L^1}=\Tr\big(\sqrt{H_0}\varrho\sqrt{H_0}\big),
$$
where the series converges in $L^1(\Rm^d)$ and $\{\lambda_p\}_{p \in \NN}$ and $\{\psi_p\}_{p \in \NN}$ are the eigenvalues and eigenvectors of $\varrho$.
\paragraph{Setting of the problem.} The local density constraint is defined in a different, more convenient form than the one given in the introduction as follows: let $\varrho$ be a density operator; for any function $\varphi \in L^\infty(\Rm^d)$, and identifying a function with its associated multiplication operator, the  density $n[\varrho]$ is uniquely defined by duality by 
$$
\int_{\Rm^d} n[\varrho] \varphi dx = \Tr \big( \varrho \varphi \big).
$$
A familiar equivalent expression is
\be \label{defcurrent}
n[\varrho ]= \sum_{p \in \Nm} \lambda_p |\psi_p|^2, 
\ee
where the series converges in $L^1(\Rm^d)$. Given a nonnegative function $n$ satisfying \fref{hypn}, the admissible set is then
\begin{align*}
\calA &= \left\{ \varrho\in \calE_+:\; n[\varrho] = n \right\}.
\end{align*}
The kinetic energy and the entropy of $\varrho \in \calE_+$ are denoted by
\begin{equation*}
E(\varrho) = \Tr \big(\sqrt{H_0} \varrho \sqrt{H_0}\big), \qquad S(\varrho)=\Tr \left(\beta(\varrho)\right),
\end{equation*}
where $\beta$ is the Boltzmann entropy $\beta(x)=x \log(x)-x$. We will state and prove our main result for such a $\beta$, and explain why it directly extends to the Fermi-Dirac entropy for instance. Setting $T=1$ to simplify notation, we write $F(\varrho)=E(\varrho)+S(\varrho)$ and consider the minimization problem
\be \label{miniprob}
\min_{\varrho \in \calA} F(\varrho).
\ee
It is proven in \cite{MP-KRM} that there exists a unique solution $\vs$ to \fref{miniprob}, that we characterize in our main result further.
 Before stating it, we need to introduce a few more notations. Consider the nonnegative potential
$$
V_\star=\frac{|\nabla \sqrt{n}|^2- n[\vs \log (\vs)] }{n},
$$
where a series expression of $n[\vs \log (\vs)]$ is given in \fref{local}. Note that $-n[\vs \log (\vs)]$ is nonnegative since the eigenvalues of $\vs$ are less than one since $\Tr(\vs)=\|n\|_{L^1}=1$. Since $\sqrt{n} \in H^1(\Rm^d)$, and we will see later that $n[\vs \log (\vs)] \in L^1(\Rm^d)$, the potential $V_\star$ is only  in $L^1(\Rm^d;ndx)$. We then define the following weighted Sobolev space
$$
H^1_{\star}(\Rm^d)=\{ u \in L^2(\Rm^d;(1+V_\star) dx): \; \nabla u \in (L^2(\Rm^d))^d\},
$$
which is complete as a closed subspace of $H^1(\Rm^d)$. Furthermore, let $\calQ_\star$ be the  quadratic form
$$
\calQ_\star(u,v)=\int_{\Rm^d} n \nabla \left(\frac{u^*}{\sqrt{n}}\right)\cdot \nabla \left(\frac{v}{\sqrt{n}}\right)dx+ \int_{\Rm^d} \left(V_\star- \frac{k[\vs] }{n} \right)u^* v dx, \qquad u,v \in H^1_\star(\Rm^d),
$$
where $k[\vs]$ is the local kinetic energy defined in \fref{local}. It is not clear at this point that $\calQ_\star$ is indeed well defined on $H^1_\star(\Rm^d)$. For this, we will see on the one hand that, and this is a consequence of the fact that $\vs$ is the minimizer of $F$,
$$
 \int_{\Rm^d} \frac{k[\vs] }{n}|u|^2dx \leq \int_{\Rm^d} n \left|\nabla \left(\frac{u}{\sqrt{n}} \right)\right|^2dx+ \int_{\Rm^d} V_\star |u|^2 dx,
$$
and on the other, after a short calculation, that
\bee
\int_{\Rm^d} n \left|\nabla \left(\frac{u}{\sqrt{n}} \right)\right|^2dx&=&\int_{\Rm^d} \left|\nabla u \right|^2dx- \int_{\Rm^d} \frac{\nabla \sqrt{n} \cdot \nabla |u|^2}{\sqrt{n}}dx+\int_{\Rm^d} \frac{|\nabla \sqrt{n}|^2}{n}|u|^2dx,
\eee
which explains why $\mathcal{Q}_\star$ is well defined on $H_\star^1(\Rm^d)$ using the Cauchy-Schwarz inequality. We will write $\calQ_\star(u)$ for $\calQ_\star(u,u)$. Let finally $\frakH$ be defined by
$$ \frakH=\left\{ \varphi \in L^2(\Rm^d): - \sum_{p \in \Nm} \log (\rho_p)|(\phi_p,\varphi)|^2 < \infty \right\},$$
where we recall  that $\{\rho_p\}_{p \in \NN}$ is the nonnegative and nonincreasing sequence of eigenvalues of $\vs$ and $\{\phi_p\}_{p \in \NN}$ its eigenvectors, which form an orthonormal basis of $L^2(\Rm^d)$. The space $\frakH$ is a Hilbert space when equipped with the inner product
$$
(u,v)_\frakH=-\sum_{p \in \Nm} \log (\rho_p)(\phi_p,u)^* (\phi_p,v).
$$
Note that $\frakH$ is well defined since we will see that $\rho_p>0$ for all $p \in \Nm$, and that $\frakH$ is the domain of self-adjointness of $\sqrt{- \log(\vs)}$.\\

We state now our main result.

\begin{theorem} \label{mainth}Let $\vs$ be the unique solution to the minimization problem \fref{miniprob} with the constraint $n$ satisfying \fref{hypn}, and denote by $\{\rho_p\}_{p\in\NN}$ and $\{\phi_p\}_{p\in\NN}$ the eigenvalues and eigenfunctions of $\vs$. Then $\vs$ is full rank, i.e. $\rho_p>0$ for all $p \in \Nm$, and
\be \label{minieig}
- \log (\rho_p) = \min_{\varphi \in \calK_p} \calQ_\star(\varphi)=\calQ_\star(\phi_p), \qquad p \in \Nm,
\ee
where
$$
\calK_p=\{ \varphi \in H^1_\star(\Rm^d): \|\varphi\|=1, \quad (\phi_q,\varphi)=0,\; q=0,\cdots,p-1\},
$$
with the convention that $\calK_0=\{ \varphi \in H^1_\star(\Rm^d): \|\varphi\|=1\}$. Moreover, denoting by $\calQ_{\star,S}$ the restriction of $\calQ_\star$ to $S=$span$\{\phi_p, p\in \Nm\}$, we have that $\calQ_{\star,S}$ is densely defined and closable, and that $- \log (\varrho_\star)$ is the unique self-adjoint operator associated with the closure $\overline{\calQ}_{\star,S}$. Finally, $H^1_\star(\Rm^d) \subset \frakH$.
\end{theorem}

Let us make a few remarks. The self-consistent eigenvalue problem  \fref{minieig} is the rigorous formulation of \fref{quantM}. Also,  while the form $\calQ_0$ obtained by setting $k[\vs]=0$ in $\calQ_\star$ can be shown to be closed in $H^1_\star(\Rm^d)$, and is therefore associated to a self-adjoint operator, we do not know if $\calQ_\star$ is closed in $H^1_\star(\Rm^d)$. This is because $-k[\vs]$ is negative and is only in $L^1(\Rm^d)$, and there does not seem to be a way to consider the term involving $k[\vs]$ as a perturbation of $\calQ_0$ with such a low regularity. We obtain though that $\calQ_\star$ is positive, and that it is closable when defined on a dense, smaller set than $H^1_\star(\Rm^d)$. Note that since $\{\phi_p\}_{p\in\NN} \subset H^1_\star(\Rm^d)$, and that $\{\phi_p\}_{p\in\NN}$ is an orthonormal basis of $L^2(\Rm^d)$, the set $H^1_\star(\Rm^d)$ is dense in $L^2(\Rm^d)$.

Theorem \ref{mainth} can be directly generalized to the Fermi-Dirac entropy $x \log(x)+(1-x) \log(1-x)$, $x \in [0,1]$. The Boltzmann and Fermi-Dirac entropies share indeed the same technical difficulties, in particular the fact that the eigenvalues of $\vs$ accumulate at zero. The Fermi-Dirac entropy has another singularity at $x=1$, which is not an issue since there is only a finite number of eigenvalues arbitrarily close to one.
\paragraph{Strategy of proof.} One of the main difficulties is to construct admissible directions in order to derive the Euler-Lagrange equations. For $\varphi \in L^2(\Rm^d)$, we will choose operators of the form (with the Dirac bra-ket notation)$$
\varrho(t)=\sqrt{\frac{n}{n(t)}} \big(\vs+ t  \ket{\varphi}\bra{\varphi}\big)\sqrt{\frac{n}{n(t)}},$$
where $n(t)=n[\vs+ t  \ket{\varphi}\bra{\varphi}]$. An issue here is to make sure that $\vs+ t  \ket{\varphi}\bra{\varphi}$ is nonnegative. This is true for any $\varphi \in L^2(\Rm^d)$ when $t \geq 0$, but is false if $t<0$, leading only to an inequality when $t \geq 0$ in the Euler-Lagrange equations and not to an equality. We will use this inequality to prove an important estimate in the derivation of \fref{minieig} and to obtain that $\vs$ is full rank. We will then replace $\ket{\varphi}\bra{\varphi}$ by $\ket{\phi_p}\bra{\phi_q}+\ket{\phi_q}\bra{\phi_p}$, which will allow us to work with negative and sufficiently small $t$ to obtain the Euler-Lagrange equations as an equality. Note that it is tempting to use operators of the form $C(t)\vs C^*(t)$ for appropriate $C(t)$ since positivity is ensured, but this does not eventually bring more information.

Another important fact is to realize that $A[\vs]$ can be written as
$$
A[\vs]= \frac{\Delta \sqrt{n}}{\sqrt{n}}+\frac{|\nabla \sqrt{n}|^2-n[\vs \log (\vs)]-k[\vs]}{n},
$$
where the first term is called the Bohm potential, and can be absorbed into the Laplacian leading to 
$$
H[\vs]=-\frac{1}{\sqrt{n}} \nabla \cdot \left(n  \nabla \left( \frac{\cdot}{\sqrt{n}}\right) \right)+V_\star -\frac{k[\vs]}{n}.
$$
It then not necessary to have some regularity on $\Delta \sqrt{n}$, and the hypotheses \fref{hypn} are sufficient.
\section{Proof of Theorem \ref{mainth}} \label{proofmainth}

The proof is divided into four parts. In the first one, we obtain important results about the differentiability (in appropriate directions) of the functional $F(\varrho)$. In the second part, we prove that the minimizer is full rank. In the third part, we derive the crucial relation $\calQ_\star(\phi_p,\phi_q)=- \log(\rho_p) \delta_{pq}$, while we conclude the proof in the fourth part.

Throughout this section, we will use the following notations. For $\varphi \in L^2(\Rm^d)$, $P_\varphi$ denotes the rank-one projector $P_\varphi= \ket{\varphi}\bra{\varphi}$. For $t\geq 0$, we consider perturbations of the minimizer of the form $\vs+t P_\varphi$, and introduce the local density
$$
n(t)=n[\vs+t P_\varphi], \qquad a(t)=\sqrt{\frac{n}{n(t)}},
$$
as well as
$$
\varrho(t)= a(t) \left( \vs+t P_\varphi \right) a(t).
$$
The operator $\varrho(t)$ is designed to belong to the admissible set $\calA$. Consider finally the weight 
$$
\omega(x)=1+V_\star(x)+\frac{k[\vs]}{n},
$$
and introduce the space
$$
H^1_\omega(\Rm^d)=\{ u \in L^2(\Rm^d;\omega dx): \; \nabla u \in (L^2(\Rm^d))^d\}.
$$
Note that we actually have $H^1_\omega(\Rm^d)=H^1_\star(\Rm^d)$, but this fact is unknown at this stage. We will need the following logarithmic Sobolev for systems proved in \cite[Corollary 18]{Dolbeault-Loss}, which holds for any $\varrho \in \calE_+$ such that $n[\varrho] \log n[\varrho] \in L^1(\Rm^d)$:
\be
\label{logsobo}
\int_{\Rm^d}n[\varrho]\log (n[\varrho]) dx\leq \sum_{p\in \Nm}\rho_p\log (\rho_p)+\frac{d}{2}\log\left(\frac{e}{2\pi d}\frac{E(\varrho)}{\Tr(\varrho)}\right)\Tr(\varrho).
\ee
Above, $\{\rho_p\}_{p \in \NN}$ denotes the set of eigenvalues of $\varrho$ and $E(\varrho)$ its kinetic energy. Since $E(\varrho)<\infty$ for any admissible $\varrho$ and $x \log (x) \leq 0$ for $x \in [0,1]$, this inequality shows that the entropy of an admissible density operator is indeed well-defined as
\be \label{ent0}
0 \leq -\sum_{p\in \Nm}\rho_p\log (\rho_p)= -\Tr(\beta(\varrho)) \leq \frac{d}{2}\log\left(\frac{e}{2\pi d}E(\varrho)\right)- \int_{\Rm^d}n\log (n) dx<\infty.
\ee
Above, we used that $\Tr(\varrho)=1=\|n\|_{L^1}$ and $n \log (n) \in L^1(\Rm^d)$.
\subsection{Preliminary results}

For $u,v \in H^1_\omega(\Rm^d)$, consider the quadratic form
 \bea \nonumber 
 \calQ_e(u,v)&=& \int_{\Rm^d} \left(-\frac{\nabla \sqrt{n} \cdot \nabla (u^* v )}{\sqrt{n}}+\frac{2|\nabla \sqrt{n}|^2 u^* v}{n} \right)dx\\
  &&+ \int_{\Rm^d} \nabla u^* \cdot \nabla v dx-  \int_{\Rm^d} \frac{k[\vs] u^* v}{n}dx,\label{defQ}
  \eea
  with the notation $\calQ_e(\varphi) \equiv \calQ_e(\varphi,\varphi)$. Note that it follows from the Cauchy-Schwarz inequality for the first term on the right above that $\calQ_e$ is indeed well defined on $H^1_\omega(\Rm^d)$. The first lemma below pertains to the kinetic energy $E(\varrho(t))$, and is proven in section \ref{proofderivE}.
  
\begin{lemma} \label{derivE} Suppose $\varphi \in H^1_\omega(\Rm^d)$. Then, for all $t \geq 0$, $\varrho(t)$ belongs to the admissible set $\calA$. Moreover, $E(\varrho(t)) \in C^1(\Rm^+)$ with
  \bea \label{derivEeq}
  \left.\frac { d E(\varrho(t))}{dt}\right|_{t=0^+}&=& \calQ_e(\varphi).
  \eea
\end{lemma}

We then consider the entropy $S(\varrho)$ for which we will need the next lemma.
\begin{lemma} \label{C1} Let $\varphi \in H^1(\Rm^d)$ with $|\varphi| \leq M \sqrt{n}$ a.e. for some $M>0$. Then, $\varrho(t) \in C^1(\Rm^+,\calJ_1)$.
\end{lemma}

Because of the singularity of $\beta'(x)$ at $x=0$ and of the particular form of $\varrho(t)$, it is difficult to justify some calculations that directly involve $S(\varrho(t))$. We therefore need to regularize and introduce, for $x \in [0,1]$ and $\eta>0$,
$$
\beta_\eta(x)= (x+\eta) \log(x+\eta)-x-\eta \log (\eta),
$$
and define
$$
S_\eta(\varrho)=\Tr(\beta_\eta(\varrho)), \qquad F_\eta(\varrho)=E(\varrho)+S_\eta(\varrho).
$$
Note that
\be \label{regbet}
\beta_\eta(x)-\beta(x)=\int_0^\eta (\log(u+x)-\log(u)) du \geq 0,
\ee
since $\log$ is an increasing function. We then obtain the following lemma.

  \begin{lemma} \label{lemS} Let $\varphi \in L^2(\Rm^d)$ with $ |\varphi| \leq M \sqrt{n}$ a.e. for some constant $M>0$. Then $S_\eta(\varrho(t)) \in C^1(\Rm^+)$, and for all $t \geq 0$,
    \be \label{posi}
    h(\eta)+S_\eta(\varrho(t))-S_\eta(\varrho(0)) \geq S(\varrho(t))-S(\varrho(0)),
    \ee
    where $h(\eta)$ is a nonnegative function independent of $t$ and $\varphi$ that goes to zero as $\eta \to 0$. Moreover,
    \be \label{entrel}
    \left.    \frac{d S_\eta(\varrho(t))}{d t} \right|_{t=s}=\Tr\big( \log(\eta+\varrho(s)) P_{\varphi a(s)}\big)-   \Tr\left( \vs \log(\eta+\varrho(s)) \frac{|\varphi|^2}{n(s)}\right).
    \ee
  \end{lemma}
   \begin{proof} That $S_\eta(\varrho(t)) \in C^1(\Rm^+)$ is a direct consequence of Lemma \ref{C1} and that, according to \cite[Lemma 5.3]{MP-JSP}, $S_\eta(\varrho)$ is differentiable at any nonnegative density operator $\varrho$ in any direction $\delta \rho \in \calJ_1$, with $D S_\eta[\varrho](\delta \rho)=\Tr(\log(\eta+\varrho) \delta \rho)$. Regarding \fref{posi}, we have first, thanks to \fref{regbet}, $S_\eta(\varrho(t))\geq S(\varrho(t))$. We  now show that $S_\eta(\vs) \to S(\vs)$ as $\eta \to 0$. Set $\eta \in (0,1/2)$, then $\beta_\eta(x) \leq 0$ for $x \in [0,1]$. 
 Then, by Fatou's lemma for series and \fref{regbet},
 \begin{align*}
   - \sum_{p \in \Nm} \beta(\rho_p)  =\sum_{p \in \Nm} \liminf_{\eta \to 0} - \beta_\eta(\rho_p) &\\
   \leq \liminf_{\eta \to 0}-\sum_{p \in \NN} &\beta_\eta(\rho_p) \leq \limsup_{\eta \to 0}-\sum_{p \in \NN} \beta_\eta(\rho_p) \leq - \sum_{p \in \Nm} \beta(\rho_p),
 \end{align*}
 which yields the desired result. Then, with $h(\eta)=S_\eta(\varrho(0))-S(\varrho(0))$ and $S_\eta(\varrho(t))\geq S(\varrho(t))$, we have
 $$
  h(\eta)+S_\eta(\varrho(t))-S_\eta(\varrho(0))=S_\eta(\varrho(0))-S(\varrho(0))+S_\eta(\varrho(t))-S_\eta(\varrho(0)) \geq S(\varrho(t))-S(\varrho(0)),
 $$
which proves \fref{posi}. Finally,
\bee
\left.    \frac{d S_\eta(\varrho(t))}{d t} \right|_{t=s} &=&\Tr \Big( \log (\eta+\varrho(s)) \partial_t \varrho(s)\Big)\\
&=&\Tr \Big( \log (\eta+\varrho(s))P_{\varphi a(s)} \Big)+   \Tr \Big( \log (\eta+\varrho(s)) \varrho(s) b(s)\Big), 
\eee
where $b(s)=- |\varphi|^2/n(s)$. Above, we used the cyclicity of the trace and $|b(s)| \leq |\varphi|^2/n \in L^\infty(\Rm^d)$, which  proves \fref{entrel}.
\end{proof}

\begin{remark}
  Note that both terms in \fref{entrel} are finite since $\varrho(s)$ and $P_{\varphi a(s)}$ are trace class, $\log(\eta+\varrho(s))$ is bounded, and $|\varphi|^2/n(s) \leq |\varphi|^2/n \in L^\infty(\Rm^d)$ by assumption. 
\end{remark}

\subsection{The minimizer is full rank}

We have the following proposition.  
\begin{proposition} \label{fullrank} The minimizer $\vs$ is full rank, that is $\rho_p >0$ for all $p \in \Nm$.
  \end{proposition}
  \begin{proof} We prove the result by contradiction, in the spirit of \cite{MP-JSP}, section 5, by differentiating in a direction related to a nonzero eigenfunction in the kernel of $\vs$. Compared to \cite{MP-JSP}, there are complications though since on the one hand, there is the additional term $h(\eta)$ in \fref{posi} that needs to be handled carefully, and on the other admissible directions do not admit as simple expressions as in \cite{MP-JSP}.

\noindent\textbf{Step 1:} Assume first that the kernel of $\vs$ is not $\{0\}$, and consider an orthonormal basis $\{\psi_p\}_{p\in I}$ of $\mbox{Ker}\,\vs$ ($I$ may be empty, finite or infinite, and we write $|I|$ for its cardinal). Then, we denote by $(\rho_p)_{1\leq p\leq N}$ the nonincreasing sequence of nonzero eigenvalues of $\vs$ (here $N$ is finite or not), associated to the orthonormal family of eigenfunctions $(\phi_p)_{1\leq p\leq N}$. We thus obtain a Hilbert basis $\{(\psi_p)_{1 \leq p \leq |I|}, (\phi_p)_{1\leq p\leq  N}\}$ of $L^2(\Rm^d)$. Pick then for instance $\psi_1$, that we denote for simplicity by $\phi$. Having little information about its regularity  (we only know it is in $L^2(\Rm^d)$), we need to regularize it in order to apply Lemmas \ref{derivE} and \ref{lemS}. Let then $\varphi_\eps=\phi_\eps (\frac{n}{n+\eps |\phi_\eps|^2})^{1/2}$, where $\phi_\eps \in C_c^\infty(\Rm^d)$ and $\phi_\eps \to \phi$ in $L^2(\Rm^d)$. We verify that $\phi_\eps \in H^1_\omega(\Rm^d)$. First of all,
\begin{align*}
\int_{\mathbb{R}^d} |\varphi_{\varepsilon}(x)|^2 \omega(x) dx =& \int_{\mathbb{R}^d} |\phi_{\varepsilon}(x)|^2 \frac{n(x)}{n(x) + \varepsilon |\phi_{\varepsilon}(x)|^2} dx
\\ &+ \int_{\mathbb{R}^d} |\phi_{\varepsilon}(x)|^2 \frac{|\nabla \sqrt{n(x)}|^2 - n[\varrho_{\star}\log(\varrho_{\star})] - k[\varrho_{\star}]}{n(x) + \varepsilon |\phi_{\varepsilon}(x)|^2} dx
\\[2mm] \leq & \|\phi_\eps\|^2 + \varepsilon^{-1} \left( \|\nabla\sqrt{n}\|^2 - \Tr(\varrho_{\star}\log(\varrho_{\star})) + \|k[\vs]\|_{L^1} \right),
\end{align*}
which is finite thanks to \fref{ent0} and the fact that $\|k[\vs]\|_{L^1}=E(\varrho_{\star})<\infty$.
Moreover, we have
\begin{align}
\nabla \varphi_{\varepsilon} &= \left(\frac{n}{n+\eps |\phi_\eps|^2}\right)^{1/2}\nabla \phi_{\varepsilon} + \phi_{\varepsilon} \nabla \sqrt{n}\left( \frac{1}{(n+\varepsilon|\phi_{\varepsilon}|^2)^{1/2}} -  \frac{n}{(n+\varepsilon|\phi_{\varepsilon}|^2)^{3/2}}\right)\nonumber
\\ &\hspace{1em}- \varepsilon\frac{\phi_{\varepsilon} \sqrt{n}\Re(\phi_{\varepsilon}^* \nabla \phi_{\varepsilon} )}{(n+\varepsilon|\phi_{\varepsilon}|^2)^{3/2}},\label{eq:derivativephieps}
\end{align}
leading to the estimate
\begin{equation*}
|\nabla \varphi_{\varepsilon}| \leq 2(|\nabla \phi_{\varepsilon}| + \varepsilon^{-1/2} |\nabla \sqrt{n}|).
\end{equation*}
This yields $\varphi_\eps \in H^1_\omega(\Rm^d)$. Note that we also have $ |\varphi_\eps| \leq \eps^{-1} \sqrt{n}$. Consider now
$$
\varrho_\eps(t)= a_\eps(t) \left( \vs+t P_{\varphi_\eps} \right) a_\eps(t),
$$
where
    $$
n_\eps(t)=n[\vs+t P_{\varphi_\eps}], \qquad a_\eps(t)=\sqrt{\frac{n}{n_\eps(t)}}.
$$
According to Lemma \ref{derivE}, $\varrho_\eps$ is admissible and $\varrho_\eps \in C^1(\Rm^+,\calJ_1)$. As a consequence, $\varrho_\eps(t) \to \varrho_\eps(0)=\vs$ in $\calJ_1$ as $t\to 0^+$. Adapting Lemma \ref{convreg} further, item (iv), we can choose the eigenbasis of $\varrho_\eps(t)$ and that of $\vs$ in such a way that the eigenvectors converge to one another in $L^2(\Rm^d)$ as $t \to 0^+$. We then pick $\{(\psi_p)_{1 \leq p \leq |I|}, (\phi_p)_{1\leq p\leq  N}\}$ to be this very basis for $\vs$, and we denote by $\{(\psi_p(t))_{1 \leq p \leq N(t)}\}$ that of $\varrho_\eps(t)$ ($N(t)$ can be finite or not). Let $\{\rho_p(t)\}_{1 \leq p \leq N(t)}$ be the eigenvalues associated with $\{\psi_p(t)\}_{1 \leq p \leq N(t)}$. We suppose that $\psi_1(t)$ is the eigenvector of $\varrho_\eps(t)$ converging to $\phi$, and we denote for simplicity $\psi_1(t)\equiv \phi(t)$. As a consequence,  adapting Lemma \ref{convreg} (iii) yields $\rho_1(t) \to 0$ as $t\to 0^+$, which we will use below. 

\noindent\textbf{Step 2:} According to \fref{entrel}, we find
\bee
\left.    \frac{d S_\eta(\varrho_\eps(t))}{d t} \right|_{t=s} 
&=&\Tr \Big( \log (\eta+\varrho_\eps(s))P_{\varphi_\eps a_\eps(s)} \Big)-\int_{\Rm^d} \frac{n[\varrho_\eps(s) \log (\eta+\varrho_\eps(s)) ]}{n(s)} |\varphi_\eps|^2 dx. 
\eee
Let $P_0(t)$ the set of indices such that $\rho_p(t) \geq 1-\eta$ for $ p \in P_0(t)$. Then, since $\log(\eta+x) \geq 0$ for $x \geq 1-\eta$, and since $\log(\eta+x) \leq 0$ for $0 \leq x \leq 1-\eta$,
\begin{align*}
  \Tr \Big( \log (\eta+\varrho_\eps(s))P_{\varphi_\eps a_\eps(s)} \Big) \leq &|(\varphi_\eps a_\eps(s),\phi(s))|^2\log (\eta+\rho_1(s))\\
  &+\sum_{p \in P_0(t)}|(a_\eps(s)\varphi_\eps,\phi_p(s))|^2\log(\rho_p(s)+\eta).
\end{align*}
Moreover, since $|\varphi_\eps|^2 \leq \eps^{-1} n$ , since $n(s) \geq n$ for all $s \geq 0$, and $x \log (x) \leq 0$ on $[0,1]$, and $-\log(\eta+x)$ is decreasing on $[0,1]$, we find
$$-\int_{\Rm^d} \frac{n[\varrho_\eps(s) \log (\eta+\varrho_\eps(s)) ]}{n(s)} |\varphi_\eps|^2 dx \leq - \eps^{-1} \Tr\big( \varrho_\eps(s) \log (\varrho_\eps(s)) \big).
$$
Choosing $\eta \in (0,1)$, we then obtain, since $\rho_p(t) \leq 1$, 
\bee
\left.    \frac{d S_\eta(\varrho_\eps(t))}{d t} \right|_{t=s} &\leq& |(\varphi_\eps a_\eps(s),\phi(s))|^2\log (\eta+\rho_1(s))+\log(2)\sum_{p \in P_0(t)}|(a_\eps(s)\varphi_\eps,\phi_p(s))|^2\\
&&- \eps^{-1} \Tr\big( \varrho_\eps(s) \log (\varrho_\eps(s)) \big). 
\eee
The bound $a_\eps(s) \leq 1$ for $s \geq 0$ shows that the second term on the right can be bounded by $\eps^{-1} \log(2)\|\varphi_\eps\|^2$. Regarding the third term, the logarithmic Sobolev inequality \fref{logsobo} yields, since $E(\varrho_\eps(t))\in C^1(\mathbb{R}^+)$ according to Lemma \ref{derivE},
\bee
0\leq - \Tr\big( \varrho_\eps(s) \log (\varrho_\eps(s)) \big) &\leq& -\int_{\Rm^d}n\log (n)  dx+\frac{d}{2}\log\left(\frac{e}{2\pi d}E(\varrho_\eps(s))\right)\\
&\leq & C+\frac{d}{2}\log\left(\frac{e}{2\pi d}\max_{s \in [0,1]} E(\varrho_\eps(s))\right)=C_{1,\eps}.
\eee
For any $t\in[0,1]$, we have therefore arrived at
$$
S_\eta(\varrho_\eps(t))-S_\eta(\varrho_\eps(0)) \leq \int_0^t |(\varphi_\eps a_\eps(s),\phi(s))|^2\log (\eta+\rho_1(s))ds + t C_{2,\eps}.
$$
for another constant $C_{2,\eps}$. Furthermore, since $E(\varrho_\eps(t))\in C^1(\mathbb{R}^+)$, there exists a constant $C_{3,\eps}$ such that
$$
E(\varrho_\eps(t))-E(\varrho_\eps(0)) \leq C_{3,\eps} t, \qquad \forall t\in [0,1].
$$
Gathering the previous estimates, we find, for $t\in[0,1]$,
$$
F_\eta(\varrho_\eps(t))-F_\eta(\varrho_\eps(0))+h(\eta) \leq \int_0^t |(\varphi_\eps a_\eps(s),\phi(s))|^2\log (\eta+\rho_1(s))ds + t C_{4,\eps}+h(\eta),
$$
for a new constant $C_{4,\eps}$. 

\noindent\textbf{Step 3:} We will show that we can choose $t$ and $\eta$ such that the right hand side is negative. First, write
$$
(\varphi_\eps,\phi)=(\phi_\eps,\phi)+\left( \left(\left(\frac{n}{n+\eps |\phi_\eps|^2}\right)^{1/2}-1\right)\phi_\eps,\phi\right).
$$
Since $\phi_\eps \to \phi$ in $L^2(\Rm^d)$ as $\eps \to 0$, the first term on the right above converges to one and the second to zero. We then choose $\eps$ sufficiently small so that $|(\phi_\eps,\phi)|>1/2$. Also, since 
$$
|(\varphi_\eps a_\eps(s),\phi(s))|^2 \to |(\varphi_\eps ,\phi)|^2\geq \frac{1}{4} \qquad \textrm{as} \quad s \to 0,
$$
there is an $s_0(\eps)>0$ such that,  for all $s\in [0,s_0(\eps)]$,
$$|(\varphi_\eps a_\eps(s),\phi(s))|^2 \geq \frac{1}{8}.
$$
Besides, since  $\rho_1(s) \to 0$ as $s \to 0$, there exists  $s_1(\eps,\delta)>0$ such that $0 \leq \rho_1(s) \leq \delta$, for all $s\in [0,s_1(\eps,\delta)]$.
Finally, set $\eta_0$ and $\delta$ sufficiently small so that
$$
\frac{1}{8} \log(\eta_0 +\delta) +1+C_{4,\eps}<0,
$$
and choose $\eta_1$ such that $h(\eta) \leq \min(s_0(\eps),s_1(\eps,\delta),1)$ for $\eta \leq \eta_1$ (we recall that such a $\eta_1$ exists since $h(\eta)\to 0$ as $\eta \to 0$). Then, for $\eta \leq \min(\eta_0,\eta_1)$, since $\log(x)$ is increasing, 
\begin{align*}
  \int_0^{h(\eta)}\Big( |(\varphi_\eps a_\eps(s),\phi(s))|^2&\log (\eta+\rho_1(s)) +1+C_{3,\eps}\Big)ds\\
  &\leq \int_0^{h(\eta)}\left( \frac{1}{8}\log (\eta+\delta) +1+C_{4,\eps}\right)ds<0.
\end{align*}
As a consequence, using \fref{posi}, for all $t \in [0, h(\eta)]$,
$$0 >h(\eta) +F_{\eta}(\varrho_\eps(t))-F_{\eta}(\varrho_\eps(0)) \geq F(\varrho_\eps(t))-F(\varrho_\eps(0)),$$
which contradicts the fact that $\vs$ is the unique minimizer of $F$. Hence, the kernel of $\vs$ is $\{0\}$, and the proof is complete.
\end{proof}

  \subsection{Euler-Lagrange equations}
  We prove here the relation $$\calQ_\star(\phi_p,\phi_q)=- \log(\rho_p) \delta_{pq}, \quad p,q \in \Nm.$$ In the previous section, we were able to use an arbitrary test function $\varphi$ in the perturbation since we only considered positive values for $t$. This ensured the positivity of $\varrho(t)$, with the drawback of only yielding an inequality in Euler-Lagrange equations (see e.g.  \fref{Eulineq}). This was enough though to prove that the minimizer is full rank. In order to obtain an equality in the Euler-Lagrange equations, we need to consider negative values of $t$ as well, which limits the choice of the test functions since $\varrho(t)$ has to be positive. We will choose below test functions related to the eigenfunctions $\phi_p$, for which the positivity of the perturbation holds. 

  We will need once again to regularize to justify the calculations. While we got away in the previous section with only regularizing the entropy term (this was justified by \fref{posi}), we need here to regularize as well the minimizer in order to obtain properly the Euler-Lagrange equation. Consider then the problem
  $$
  \min_{\varrho \in \calA} F_\eta(\varrho).
  $$
As \fref{miniprob}, the above problem admits a unique solution denoted by $\varrho_\eta$, with eigenvalues and eigenvectors $\{\rho_{p,\eta}\}_{p \in \Nm}$ and $\{\phi_{p,\eta}\}_{p \in \Nm}$.

\paragraph{Step 1: Euler-Lagrange equations for the regularized problem.} For $p$ and $q$ given, consider the operator
  $$
  P_\eta= \ket{\phi_{p,\eta}}\bra{\phi_{q,\eta}}+ \ket{\phi_{q,\eta}}\bra{\phi_{p,\eta}},
  $$
  that will be used to define a new direction of perturbation. It is not difficult to see that $\varrho_\eta+ tP_\eta$ is positive for $t > -\min(\rho_{p,\eta},\rho_{q,\eta})$. It is also clear that $\varrho_\eta+ tP_\eta$ is self-adjoint and trace class, and that $n_\eta(t):=n[\varrho_\eta+ tP_\eta]=n+2 t\Re (\phi_{p,\eta} \phi_{q,\eta}^*)$. We define then
  $$\varrho_\eta(t)= a_\eta(t) \big(\varrho_\eta+ tP_\eta \big) a_\eta(t), \qquad a_\eta(t)=\sqrt{\frac{n}{n_\eta(t)}}.$$
 The lemma below, proved in Section \ref{proofrhoa}, shows  that $\varrho_\eta(t)$ is in fact admissible for appropriate $t$.
  \begin{lemma} \label{rhoA}
  Let $t_0 = \min(\rho_{p,\eta},\rho_{q,\eta})/2$. Then $\varrho_\eta(t) \in \calA$ for any $t\in[-t_0,t_0]$.
    \end{lemma}
  
 We want to apply Lemma \ref{lemS} next to obtain the Euler-Lagrange equations. For this, we will see in Lemma \ref{convreg} that $\rho_{j,\eta} \to \rho_{j}$ as $\eta \to 0$, for all $j \in \Nm$. As a consequence, since $\rho_j>0$ for all $j \in \Nm$ according to Proposition \ref{fullrank}, there exists $\eta_0>0$ such that $\rho_{p,\eta}>0$ and $\rho_{q,\eta}>0$ for all $\eta \in (0,\eta_0)$. Since $n[\varrho_\eta]=n$, this leads to
  \be \label{boundphi}
  |\phi_{p,\eta}| \leq \rho_{p,\eta}^{-1/2} \sqrt{n}, \qquad a.e.,
  \ee
  with a similar estimate for $\phi_{q,\eta}$. An easy adaptation of Lemmas \ref{C1} and \ref{lemS} shows then that $\varrho_\eta \in C^1([-t_0,t_0],\calJ_1)$, and that
    $$
\left.    \frac{d S(\varrho_\eta(t))}{d t} \right|_{t=0} = 2 \log (\eta+\rho_{p,\eta}) \delta_{pq}- 2 \Re \int_{\Rm^d} \frac{n[\varrho_{\eta} \log(\eta+ \varrho_\eta)]}{n} \phi_{p,\eta} \phi_{q,\eta}^* dx.
$$

We consider now the kinetic energy term. Since $\varrho_\eta \in \calE$, we have the relation
  \be \label{KE}
  \|\nabla \phi_{p,\eta}\|^2 \leq \rho_{p,\eta}^{-1} E(\varrho_\eta)<\infty,
  \ee
  with a similar estimate for $\phi_{q,\eta}$. With \fref{boundphi}, this shows that $\phi_{p,\eta}$ and $\phi_{q,\eta}$ belong to $H^1_\omega(\Rm^d)$. Adapting Lemma \ref{derivE} then yields 
$$
  \left.\frac { d E(\varrho_\eta(t))}{dt}\right|_{t=0}= 2 \Re \calQ_e(\phi_{p,\eta},\phi_{q,\eta}).
$$
Finally, since $\varrho_\eta$ is the minimizer, the derivative of $F_\eta(\varrho_\eta(t))$ at $t=0$ vanishes, and we find, gathering the above results,
\be \label{real}
2 \Re \calQ_\eta(\phi_{p,\eta},\phi_{q,\eta})=-2 \log (\eta+\rho_{p,\eta}) \delta_{pq},
\ee
where
\be
\calQ_{\eta}(\varphi)
=\int_{\Rm^d} n \left|\nabla \left(\frac{\varphi}{\sqrt{n}}\right)\right|^2dx+ \int_{\Rm^d} \frac{|\nabla \sqrt{n}|^2- k[\varrho_\eta] - n[\varrho_\eta \log (\eta+\varrho_\eta)] }{n}|\varphi|^2dx.
\ee
Note that direct calculations show that $\calQ_\eta$ is actually equal to 
\bee
\calQ_{\eta}(\varphi)&=&\calQ_e(\varphi)-\int_{\Rm^d} \frac{n[\varrho_\eta \log (\eta+\varrho_\eta)] }{n}|\varphi|^2dx\\
&=&\int_{\Rm^d} \left|\nabla \varphi\right|^2dx- \int_{\Rm^d} \frac{\nabla \sqrt{n} \cdot \nabla |\varphi|^2}{\sqrt{n}}dx\\
&&+\int_{\Rm^d} \frac{2|\nabla \sqrt{n}|^2- k[\varrho_\eta] - n[\varrho_\eta \log (\eta+\varrho_\eta)] }{n}|\varphi|^2dx.
\eee
Replacing $P_\eta$ in the definition of $\varrho_\eta(t)$ by $i\ket{\phi_{p,\eta}}\bra{\phi_{q,\eta}}-i\ket{\phi_{q,\eta}}\bra{\phi_{p,\eta}}$, we find that $n[\varrho_\eta(t)]=n-2t \Im(\phi_{p,\eta} \phi_{q,\eta}^*)$. Repeating the above procedure, we find 
$$
2 \Im \calQ_\eta(\phi_{p,\eta},\phi_{q,\eta})=0,
$$
which, together with \fref{real}, yields
$$
\calQ_\eta(\phi_{p,\eta},\phi_{q,\eta})=-\log (\eta+\rho_{p,\eta}) \delta_{pq}.
$$
Note that we are able to obtain this equality since $\varrho_\eta$ is the minimizer of $F_\eta$, had we just regularized the entropy term we would have only obtained an inequality of the form  $\calQ_\eta(\phi_{p},\phi_{q})\geq -\log (\eta+\rho_{p}) \delta_{pq}$.

\paragraph{Step 2: Passing to the limit.} Choose for instance for $\eta$ the sequence $\eta\equiv\eta_\ell=1/\ell$ which converges to zero as $\ell \to \infty$. The following lemma, proved in section \fref{proofconvreg}, lists the convergence properties of $\varrho_\eta$.

\begin{lemma} \label{convreg}Let $\varrho_\ell:=\varrho_{\eta_\ell}$. Then, as $\ell \to \infty$: 
\begin{itemize}
\item[(i)] $\varrho_\ell$ converges to $\vs$ in $\calJ_1$.
\item[(ii)] $\sqrt{H_0} {\varrho_\ell} \sqrt{H_0}$ converges to $\sqrt{H_0} {\vs} \sqrt{H_0}$ in $\calJ_1$, and 
$\sqrt{H_0} \sqrt{\varrho_\ell}$ converges to $\sqrt{H_0} \sqrt{\vs}$ in $\calJ_2$.
\item[(iii)] $\forall p \in \Nm$, $\rho_{p,\ell}$ converges to $\rho_p$, where $\{\rho_{p,\ell}\}_{p \in \NN}$ are the eigenvalues of $\varrho_\ell$ and $\{\rho_{p}\}_{p \in \NN}$ those of $\varrho_\star$.
\item[(iv)] there exist a sequence of orthonormal eigenbasis 
$\{\phi_{p,\ell}\}_{p\in \Nm}$ of $\varrho_{\ell}$ and an orthonormal eigenbasis $\{\phi_{p}\}_{p\in \Nm}$ of $\vs$ such that, $\forall p\in \Nm$, 
$$\lim_{\ell\to +\infty}\|\phi_{p,\ell}-\phi_p\|=0\quad\textrm{and}\quad \lim_{\ell\to +\infty}\|\sqrt{\rho_{p,\ell}}\,\nabla \phi_{p,\ell}-\sqrt{\rho_{p}}\, \nabla \phi_p\|=0.$$
\item[(v)] $\beta_{\eta_\ell}(\varrho_\ell)$ converges to $\beta(\vs)$ in $\calJ_1$.
\item[(vi)] $\varrho_\ell \log (\eta_\ell+\varrho_\ell)$ converges to $\vs \log (\vs)$ in $\calJ_1$. 
\end{itemize}
\end{lemma}

Following the above lemma, we suppose that the basis of eigenvectors $\{\phi_{p}\}_{p\in \Nm}$ introduced in the previous sections is the one of item (iv). We have then the

\begin{proposition} \label{eigeq} For all $p,q \in \NN$,

$$  \calQ_\star(\phi_p,\phi_q)= - \delta_{pq} \log (\rho_q). $$ 
 Note that $\log(\rho_p)$ is well defined for all $p \in \NN$ since $\rho_p>0$ according to Proposition \ref{fullrank}.
  \end{proposition}

  \begin{proof} Set first $\rho_{p,\ell}:=\rho_{p,\eta_\ell}$ and $\phi_{p,\ell}:=\phi_{p,\eta_\ell}$. Then, since $\rho_p>0$ for all $p\in \NN$, Lemma \ref{convreg} (iii) shows that $\log(\eta_\ell+\rho_{p,\ell}) \to \log(\rho_p)$ for all $p \in \Nm$ as $\ell \to \infty$. Consider now $\sqrt{\rho_{p,\ell} \rho_{q,\ell}}\calQ_{\eta_\ell}(\phi_{p,\ell},\phi_{q,\ell})$, that we split into five terms $Q_i^\ell$ defined below, $i=1,\cdots, 5$. Then, by Lemma \ref{convreg} (iii)-(iv), as $\ell \to \infty$,
$$
Q_1^\ell:=\sqrt{\rho_{p,\ell} \rho_{q,\ell}}\int_{\Rm^d} \nabla \phi^*_{p,\ell} \cdot \nabla\phi_{q,\ell} dx \to \sqrt{\rho_{p} \rho_{q}}\int_{\Rm^d} \nabla \phi^*_{p} \cdot \nabla\phi_{q} dx.
$$
Furthermore,
\bee
Q_2^\ell&:=&- \sqrt{\rho_{p,\ell} \rho_{q,\ell}}\int_{\Rm^d} \frac{\nabla \sqrt{n} \cdot \nabla (\phi^*_{p,\ell} \phi_{q,\ell})}{\sqrt{n}}dx\\
&=&- \sqrt{\rho_{p,\ell} \rho_{q,\ell}}\int_{\Rm^d} \frac{\nabla \sqrt{n} \cdot \left(\nabla (\phi^*_{p,\ell}) \phi_{q,\ell}+ \nabla (\phi_{q,\ell}) \phi^*_{p,\ell}\right)}{\sqrt{n}}dx\\
&:=&Q_{2,1}^\ell+Q_{2,2}^\ell.
\eee
We write
\bea \label{Q2}
Q_{2,1}^\ell&=&- \sqrt{\rho_{p} \rho_{q}} \int_{\Rm^d} \frac{\nabla \sqrt{n} \cdot \nabla \phi^*_{p} \phi_{q}}{\sqrt{n}}dx\\
\nonumber &&-\sqrt{\rho_{q,\ell}} \int_{\Rm^d} \frac{\nabla \sqrt{n} \cdot \left(\nabla (\sqrt{\rho_{p,\ell}}\phi^*_{p,\ell}-\sqrt{\rho_{p}}\phi^*_{p}) \phi_{q,\ell} \right)}{\sqrt{n}}dx\\ \nonumber
&&- \sqrt{\rho_{p} } \int_{\Rm^d} \frac{\nabla \sqrt{n} \cdot \nabla \phi^*_{p}
(\sqrt{\rho_{q,\ell}}\phi_{q,\ell}-\sqrt{\rho_{q}}\phi_q)}{\sqrt{n}}dx.
\eea 
The second term on the right is controlled by
$$
 \|\nabla \sqrt{n}\| \|\nabla (\sqrt{\rho_{p,\ell}}\phi^*_{p,\ell}-\sqrt{\rho_{p}}\phi^*_{p})\| \| \sqrt{\rho_{q_\ell}} \phi_{q,\ell} /\sqrt{n}\|_{L^\infty},
$$
and goes to zero as $\ell \to \infty$ because of Lemma \ref{convreg} (iv) and the fact that 
\be \label{bound1} \frac{\sqrt{\rho_{q,\ell}} | \phi_{q,\ell}|}{ \sqrt{n}} \leq 1, \quad a.e.
\ee
 since $n[\varrho_\ell]=n$. For the third term, we deduce from Lemma \ref{convreg} (iii)-(iv) that there is a subsequence $\{k_\ell \}_{\ell \in \NN}$ such that $\sqrt{\rho_{q,k_\ell}}\phi_{q,k_\ell}$ converges to $\sqrt{\rho_{q}}\phi_q$ a.e.. Since $\nabla \sqrt{n}$ and $\nabla \phi_p$ belong to $L^2(\Rm^d)$, and $ \sqrt{\rho_{q,k_\ell}} |\phi_{q,k_\ell}| / \sqrt{n} \leq 1$ a.e. as well as $ \sqrt{\rho_{q}} |\phi_{q}| / \sqrt{n} \leq 1$ a.e., we can invoke dominated convergence and obtain that the limit of the third term is zero. We have therefore obtained that $Q_{2,1}^{k_\ell}$ converges as $\ell \to \infty$ to the first term on the right in \fref{Q2}. The term $\calQ_{2,2}^\ell$ is handled exactly as $\calQ_{2,1}^\ell$.

We treat now the term $Q_3^\ell$ that reads
\bee
Q_3^\ell&:=&- \sqrt{\rho_{p,\ell} \rho_{q,\ell}} \int_{\Rm^d} \frac{ k[\varrho_\ell] \phi^*_{p,\ell} \phi_{q,\ell}}{n}dx
\eee
According to Lemma \ref{convreg} (ii), we can conclude that $k[\varrho_\ell]$ converges to $k[\vs]$ strongly in $L^1(\Rm^d)$, and we have, using \fref{bound1},
\bee
\lim_{\ell \to \infty }Q_3^{\ell}&:=&- \lim_{\ell \to \infty } \sqrt{\rho_{p,k'_\ell} \rho_{q,k'_\ell}} \int_{\Rm^d} \frac{ k[\vs] \phi^*_{p,\ell} \phi_{q,\ell}}{n}dx.
\eee
Proceeding in the same way as $Q^\ell_{2,1}$, with dominated convergence and \fref{bound1}, we find
\bee
\lim_{\ell \to \infty }Q_3^{k_\ell}
&=&- \sqrt{\rho_{p} \rho_{q}} \int_{\Rm^d} \frac{ k[\vs] \phi^*_{p} \phi_{q}}{n}dx.
\eee
The term
\bee
Q_4^\ell&:=&- \sqrt{\rho_{p,\ell} \rho_{q,\ell}} \int_{\Rm^d} \frac{ n [\varrho_\ell \log(\eta_\ell+\varrho_\ell)] \phi^*_{p,\ell} \phi_{q,\ell}}{n}dx
\eee
is treated exactly as $Q_3^\ell$ since $n [\varrho_\ell \log(\eta_\ell+\varrho_\ell)] \to  n [\vs \log(\vs)]$ in $L^1(\Rm^d)$ according to  Lemma \ref{convreg} (vi). Finally,
\bee
\lim_{\ell \to \infty }Q_5^{k_\ell}&:=&2\lim_{\ell \to \infty } \sqrt{\rho_{p,k_\ell} \rho_{q,k_\ell}} \int_{\Rm^d} \frac{ |\nabla \sqrt{n}|^2 \phi^*_{p,k_\ell} \phi_{q,k_\ell}}{n}dx\\
&=&2 \sqrt{\rho_{p} \rho_{q}} \int_{\Rm^d} \frac{ |\nabla \sqrt{n}|^2\phi^*_{p} \phi_{q}}{n}dx,
\eee
as an application, as earlier, of dominated convergence and \fref{bound1}. Gathering the different limits, we find
$$\sqrt{\rho_p}\sqrt{\rho_q}  \left(\calQ_\star(\phi_p,\phi_q) - \delta_{pq} \log(\rho_p)\right)=0,
$$
which ends the proof since $\rho_p$ and $\rho_q$ are strictly positive according to Proposition \ref{fullrank}. 
 \end{proof}

\paragraph{An important estimate.} The next result is central in proving the eigenvalue relation \fref{minieig}.

  \begin{proposition} \label{estimH} Let $\varphi \in H^1_\omega(\Rm^d)$. Then, 
   $$ \| \varphi \|^2_{\frakH} \leq \calQ_\star(\varphi).$$
 \end{proposition}

 \begin{proof} For $t\geq 0$ and $\varphi \in H^1_\omega(\Rm^d)$, consider the operator
$$
\varrho_\eta(t)= a_\eta(t) \left( \varrho_\eta+t P_\varphi \right) a_\eta(t), 
$$
where
$$
n_\eta(t)=n[\varrho_\eta+t P_\varphi], \qquad a_\eta(t)=\sqrt{\frac{n}{n_\eta(t)}}.
$$
We need to regularize $\varphi$ in order to have the estimate $|\varphi| \leq M \sqrt{n}$ and use Lemma \ref{lemS}. Let then $\varphi_\eps=\varphi \left(\frac{n}{n+\eps |\varphi|^2}\right)^{1/2}$. We can see that $ |\varphi_\eps| \leq \eps^{-1} \sqrt{n}$ and it follows from similar computations as in \eqref{eq:derivativephieps} that 
\begin{equation}\label{boundgrad}
|\nabla \varphi_{\varepsilon}| \leq 2\left(|\nabla \varphi| + \frac{|\nabla\sqrt{n}|}{\sqrt{n}} |\varphi|\right),
\end{equation}
which gives $\varphi_\eps \in H^1_\omega(\Rm^d)$. Denoting by $\varrho_\eps(t)$ (and dropping the dependency on $\eta_\ell$ to ease notation) the operator $\varrho_{\eta_\ell}(t)$ for $\varphi \equiv \varphi_\eps$, Lemma \ref{derivE} and \fref{entrel} then yield, since $\varrho_{\eta_\ell}$ is the minimizer, 
\be \label{Eulineq}
\left.    \frac{d F_{\eta_\ell}(\varrho_\eps(t))}{d t} \right|_{t=0^+}= \calQ_{\eta_\ell}(\varphi_\eps)+\Tr\big( \log(\eta_\ell+\varrho_{\eta_\ell}) P_{\varphi_\eps}\big) \geq 0.
\ee
We will pass to the limit in the above relation. We have first
$$
\Tr\big( \log(\eta_\ell+\varrho_{\eta_\ell}) P_{\varphi_\eps}\big) =\sum_{p\in \NN} |(\varphi_\eps,\phi_{p,\ell})|^2\log(\rho_{p,\ell}+\eta_\ell).
$$
Choose $\eta_\ell \in (0,1/2]$ and $p_0$ such that $\rho_{p} \leq 1/2$ for $p \geq p_0$. Since $\rho_{p,\ell} \to \rho_p$ by Lemma \fref{convreg} (iii), we can choose $\ell$ sufficiently large that $\rho_{p,\ell} \leq 1/2$ for $p \leq p_0$. Since $\varphi_\eps \to \varphi$ in $L^2(\Rm^d)$, it follows that, with Lemma \ref{convreg} (iii)-(iv),
$$
 \lim_{\eps \to 0} \lim_{\ell \to \infty}\sum_{p < p_0} |(\varphi_\eps,\phi_{p,\ell})|^2\log(\rho_{p,\ell}+\eta_\ell)=\sum_{p < p_0} |(\varphi,\phi_p)|^2\log(\rho_p).
$$
Note that $\log(\rho_p)$ is well defined according to Proposition \ref{fullrank}. Then, since 
$-\log(\eta_\ell+\rho_{p,\ell}) \geq 0$ for $p \geq p_0$, it follows from Fatou's lemma for series that 
$$
-\sum_{p \geq p_0} |(\varphi,\phi_p)|^2\log(\rho_p)\leq -\liminf_{\eps \to 0, \ell \to \infty} \sum_{p\geq p_0 } |(\varphi_\eps,\phi_{p,\ell})|^2\log(\rho_{p,\ell}+\eta_\ell).
$$

It remains now to pass to the limit in $\calQ_{\eta_\ell}(\varphi_\eps)$. The limit in $\eta_\ell$ is done in the exact same way as in the proof of Proposition \ref{eigeq}, we simply use $|\varphi_\eps| \leq \eps^{-1}\sqrt{n}$ instead of \fref{bound1} in order to apply dominated convergence. We then replace all terms in $\eta_\ell$ by their limit and treat now the term in $\calQ_\star(\varphi_\eps)$ involving $\nabla \varphi_\eps$. We have that $\nabla \varphi_{\varepsilon}$ is given by \eqref{eq:derivativephieps} and converges to $\nabla \varphi$ a.e. as $\eps \to 0$. With the estimate \eqref{boundgrad} and the fact that the r.h.s is in $L^2(\Rm^d)$ because $\varphi \in H^1_\omega(\Rm^d)$, we can invoke dominated convergence and pass to the limit and obtain
$$
\lim_{\eps \to 0} \int_{\Rm^d} |\nabla \varphi_\eps|^2dx=\int_{\Rm^d} |\nabla \varphi|^2dx. 
$$
It remains to treat
$$  \int_{\Rm^d} \left(-\frac{\nabla \sqrt{n} \cdot \nabla |\varphi_\eps|^2}{\sqrt{n}}+\frac{2|\nabla \sqrt{n}|^2-k[\vs]}{n}  |\varphi_\eps|^2\right)dx. 
$$
Using that $\nabla \varphi_\eps \to \nabla \varphi$ and $\varphi_\eps \to \varphi$ both a.e., together with \fref{boundgrad} and $|\varphi_\eps| \leq |\varphi|$, we can use dominated convergence to pass to the limit above and obtain the desired result. This ends the proof.
  \end{proof}

  \subsection{Conclusion and proof of the main theorem}

 We have already obtained in Proposition \ref{fullrank} that $\vs$ is full rank. We prove now relation \fref{minieig}. According to Proposition \ref{estimH}, we have, since $\rho_p<1$ for all $p \in \NN$,
\be \label{below}
- \log (\rho_0)\|\varphi\|^2 \leq \|\varphi\|^2_\frakH \leq \calQ_\star(\varphi), \qquad \forall \varphi \in H^1_\omega(\Rm^d).
\ee
This shows in particular that $Q_\star$ is nonnegative and as a consequence
$$
\int_{\Rm^d} \frac{ k[\vs] }{n} |\varphi|^2 dx\leq \int_{\Rm^d} n \left| \nabla \left(\frac{\varphi}{\sqrt{n}}\right)\right|^2dx+ \int_{\Rm^d} V_\star |\varphi|^2 dx,
$$
which yields that $H^1_\omega(\Rm^d)=H_\star^1(\Rm^d)$. Besides, we deduce from \fref{below} that
$$
- \log (\rho_0) \leq \inf_{\varphi \in H^1_\star, \|\varphi\|=1} Q_\star(\varphi).
$$
According to Proposition \ref{eigeq}, we have $Q_\star(\phi_0)=- \log (\rho_0)$, and therefore the above infimum is attained at $\phi_0$. At any order $p>1$, we have, for any $\varphi \in \calK_p$, 
$$
- \log (\rho_p)\|\varphi\|^2 \leq \|\varphi\|^2_\frakH \leq \calQ_\star(\varphi) \qquad \textrm{so that} \qquad
- \log (\rho_p) \leq \inf_{\varphi \in \calK_p} Q_\star(\varphi),
$$
and, according to Proposition \ref{eigeq}, the infimum is attained at $\phi_p$. This proves \fref{minieig}. Consider now $\calQ_{\star,S}$, which is densely defined since $\{\phi_p\}_{p \in \Nm}$ is an orthonormal basis of $L^2(\Rm^d)$. Following Proposition \ref{eigeq}, we have
$$
\calQ_{\star,S}(u,v)=(u,v)_{\frakH}, \qquad \forall u,v \in S.
$$
Since $\frakH$ is complete, this shows that $\calQ_{\star,S}$ is closable and that
$$
\overline{\calQ}_{\star,S}(u,v)=(u,v)_{\frakH}=(u,- \log(\vs) v), \qquad \forall u \in \frakH, \; \forall v \in D(- \log(\vs)),
$$
where 
$$D(- \log(\vs))=\left\{ \varphi \in L^2(\Rm^d): \sum_{p \in \Nm} (\log (\rho_p))^2 |(\phi_p,\varphi)|^2 < \infty \right\}.$$
Finally, the fact that $H_\star^1(\Rm^d) \subset \frakH$ is a consequence of Proposition \ref{estimH}. This ends the proof of the theorem.

\section{Proofs of some lemmas} \label{lemproof}

\subsection{Proof of Lemma  \ref{derivE}} \label{proofderivE} 

\noindent\textbf{Step 1:} We show first that $\varrho(t) \in \calA$ for all $t\geq 0$. It is direct to see that  $\varrho(t)$ is nonnegative since $t \geq 0$, and that it is trace class as products of the trace class operator $\vs+t P_\varphi$ and the bounded multiplication operator by $a(t) \leq 1$. It then follows that 
\begin{equation*}
n[\varrho(t)] = a(t)^2(n + t |\varphi|^2) = n,
\end{equation*}
\textit{i.e.} $\varrho(t)$ satisfies the constraint $n[\varrho(t)]=n$

\noindent\textbf{Step 2:} We want to prove that $\varrho(t)\in\mathcal{E}$. For this, we write, by linearity of the trace,
  $$
  E(\varrho(t))= E(a(t) \vs a(t))+ t E(a(t) P_\varphi a(t)),
  $$
  and use the following result, which is just a consequence of the definition of Hilbert-Schmidt operators:
  \begin{lemma}
  Let $\sigma $ be a density operator and $A$ be a self-adjoint operator. Then $A \sigma A \in \calJ_1$ if and only if $A \sigma^{1/2} \in \calJ_2$.
  \end{lemma}
  According to \cite[Theorem 6.22, item (g)]{RS-80-I}, $A \sigma^{1/2} \in \calJ_2$ holds provided there is an othonormal basis $\{e_p\}_{p \in \Nm}$ of $L^2(\Rm^d)$ such that
  $$
  \sum_{p \in \Nm} \| A \sigma^{1/2} e_p\|^2 < \infty.
  $$
  We use this result with $A=\sqrt{H_0}$ and $\sigma=a(t) \vs a(t)$ as follows. Noticing first that $$
  \left\| \sqrt{H_0} \left(a(t) \vs^{1/2} \phi_p \right) \right\|= \left\| \nabla \left(a(t) \vs^{1/2} \phi_p \right)\right\|,$$
  we have
  \begin{align} \label{EE1}
  E(a(t) \vs a(t))&=\sum_{p \in \Nm}  \left\| \sqrt{H_0} \left(a(t) \vs^{1/2} \phi_p \right) \right\|^2 = \sum_{p \in \Nm} \rho_p  \| \nabla \left( a(t) \phi_p \right) \|^2\\
 \nonumber &= \int_{\Rm^d} \sum_{p \in \Nm} \rho_p \left( |\nabla a(t)|^2 |\phi_p|^2+ 2 a(t) \nabla a(t) \cdot \Re(\phi_p^* \nabla \phi_p ) + a(t)^2 |\nabla \phi_p|^2\right) dx\\
  \nonumber &= \int_{\Rm^d} \left( n |\nabla a(t)|^2+ a(t )\nabla a(t) \cdot \nabla n+ a(t)^2 k[\vs]\right)dx.
  \end{align}
  Above, the exchange of the integral and the summation is justified since the integrand is positive. We also used the fact that
  $$
    n= \sum_{p \in \Nm} \rho_p |\phi_p|^2, \quad 
\nabla n=2 \sum_{p \in \Nm} \rho_p  \Re (\phi_p^* \nabla \phi_p), \quad\textrm{and}\quad  k[\vs]= \sum_{p \in \Nm} \rho_p |\nabla \phi_p|^2,
  $$
  with convergence of the series in $L^1(\Rm^d)$ and almost everywhere. Regarding $E(a(t) P_\varphi a(t))$, we have
  $$
  E(a(t) P_\varphi a(t))=\| \nabla (a(t) \varphi) \|^2 = \| \varphi \nabla a(t)\|^2 +  \| a(t)\nabla \varphi \|^2 + 2\Re (\varphi \nabla a(t) , a(t)\nabla \varphi),
  $$
  which also gives 
  \begin{align*}
  \int_{\Rm^d} &\left( n |\nabla a(t)|^2+ a(t )\nabla a(t) \cdot \nabla n\right)dx + t  E(a(t) P_\varphi a(t))
  \\ &=  \int_{\Rm^d} \left( (n(t) \nabla a(t)+ a(t) \nabla n(t))\cdot \nabla a(t) + ta(t)^2|\nabla\varphi|^2 \right)dx.
  \end{align*}
Thus, we find the expression
  \begin{equation*}
  E(\varrho(t))=  \int_{\Rm^d} \left( (n(t) \nabla a(t)+ a(t) \nabla n(t))\cdot \nabla a(t)+ a(t)^2 (k[\vs]+t |\nabla \varphi|^2) \right)dx.
  \end{equation*}
  Note that it is crucial to express $E(\varrho(t))$ in terms of the moments of $\vs$ in order to exploit the fact that $\varphi \in H^1_\omega(\Rm^d)$. Working directly with $\vs$ and $a(t)$ as operators would make it difficult to justify the calculations leading to \fref{derivEeq}. From there, with the relation
  $$
  \nabla a(t)= \frac{\nabla \sqrt{n}}{\sqrt{n(t)}}-\frac{a(t)\nabla\sqrt{n(t)} }{\sqrt{n(t)}},
  $$
  we find
  \begin{align*}
  (n(t) \nabla a(t) + &a(t) \nabla n(t))\cdot \nabla a(t) 
  \\ &=   \left(\sqrt{n(t)} \nabla \sqrt{n} - \sqrt{n(t)} a(t)\nabla\sqrt{n(t)} + a(t) \nabla n(t)\right)\cdot \nabla a(t)
  \\ &=   \left(\nabla \sqrt{n} - a(t)\nabla\sqrt{n(t)} + a(t) \frac{\nabla n(t)}{\sqrt{n(t)}}\right)\cdot (\nabla \sqrt{n} -a(t) \nabla\sqrt{n(t)})
  \\ &=   \left(\nabla \sqrt{n} + a(t)\nabla\sqrt{n(t)} \right)\cdot (\nabla \sqrt{n} -a(t) \nabla\sqrt{n(t)})
  \\ &=   |\nabla \sqrt{n}|^2 - a(t)^2|\nabla\sqrt{n(t)}|^2,
  \end{align*}
  leading to
  \begin{align} \label{eqE}
  E(\varrho(t)) =  \int_{\Rm^d} \left( |\nabla \sqrt{n}|^2 - a(t)^2|\nabla\sqrt{n(t)}|^2 + a(t)^2 (k[\vs]+t |\nabla \varphi|^2) \right)dx.
  \end{align}
  We remark that
  \begin{align*}
  \left|\nabla\sqrt{n(t)}\right|\leq \left|\nabla \sqrt{n}\right| + \sqrt{t} \left|\nabla \varphi\right|,
  \end{align*}
  so that $\nabla\sqrt{n(t)}\in (L^2(\mathbb{R}^d))^d$ for any $t\geq 0$. Since $a(t)$ bounded by one, $\nabla \sqrt{n}$ and $\nabla \varphi$ are in $(L^2(\Rm^d))^d$ and $k[\vs]$ is in $L^1(\Rm^d)$, it follows that  $E(\varrho(t))$ is finite and, as a consequence, that $\varrho(t) \in \calE$.
 
 \noindent\textbf{Step 3:} We show that $t\to E(\varrho(t))$ is differentiable. First of all, it is clear that $a(t)$ and $\sqrt{n(t)}$ are continuously differentiable as functions of $t$, for almost all $x$. Denote by $a_1(t,x)$ the integrand in $E(\varrho(t))$, which is then continuously differentiable as a function of $t$. With the following relations
$$
\partial_t a(t)=-\frac{a(t) |\varphi|^2}{2n(t)}
 \quad \textrm{and} \quad \partial_t \nabla \sqrt{n(t)}=\frac{\Re(\varphi^* \nabla \varphi)}{\sqrt{n(t)}}-\frac{|\varphi|^2 \nabla \sqrt{n(t)}}{2 n(t)},
$$
tedious but straightforward calculations show that
$$
|\partial_t a_1(t)| \leq C \left( \frac{|\nabla\sqrt{n}|^2}{n} |\varphi|^2 +|\nabla \varphi|^2+ \frac{k[\vs]}{n}|\varphi|^2\right).
$$
Since $\varphi\in H^1_\omega(\Rm^d)$, the function of the r.h.s above is integrable, and standard results about Lebesgue integration imply then that $E(\varrho(t)) \in C^1(\Rm^+)$.

\noindent\textbf{Step 4:} We consider now \fref{derivEeq}. Differentiating $E(\varrho(t))$ leads to
\begin{align*}
\partial_t E(\varrho(t)) =&  - 2\int_{\Rm^d}a(t) \left(  |\nabla\sqrt{n(t)}|^2 \partial_t a(t) + a(t) \nabla \sqrt{n(t)} \cdot \partial_t  \nabla \sqrt{n(t)}\right) dx\nonumber
\\    &+\int_{\Rm^d} \left(  2 a(t)\partial_t a(t) (k[\vs]+t |\nabla \varphi|^2) + a(t)^2|\nabla \varphi|^2  \right)dx\nonumber
\\ =&  \int_{\Rm^d}a(t)^2 \left( 2 \frac{|\nabla\sqrt{n(t)}|^2}{n(t)} |\varphi|^2 - \frac{\nabla \sqrt{n(t)}}{\sqrt{n(t)}} \cdot  \nabla |\varphi|^2\right) dx\nonumber
\\    &+\int_{\Rm^d}a(t)^2  \left( |\nabla \varphi|^2 - \frac{|\varphi|^2}{n(t)} \left(k[\vs]+t |\nabla \varphi|^2\right)   \right)dx.
\end{align*}
Since
\begin{equation*}
a(0) = 1\quad\textrm{and}\quad n(0) = n,
\end{equation*}
we directly deduce that
\begin{align*}
\left.    \frac{d E(\varrho(t))}{d t} \right|_{t=0} =&  \int_{\Rm^d} \left( |\nabla \varphi|^2 + 2 \frac{|\nabla\sqrt{n}|^2}{n} |\varphi|^2 - \frac{\nabla \sqrt{n}}{\sqrt{n}} \cdot  \nabla |\varphi|^2 - \frac{|\varphi|^2}{n} k[\vs]\right) dx,
\end{align*}
which is the desired result.
 \subsection{Proof of Lemma \ref{C1}.} 
 First of all,  it is clear that for any $t \geq 0$ and $x\in\mathbb{R}^d$, $$a(t,x)=\frac{\sqrt{n(x)}}{\sqrt{n(x)+t |\varphi(x)|^2}}$$ is bounded by one. Furthermore, we have
\begin{equation*}
\partial_t a(t)=-\frac{a(t) |\varphi|^2}{2 n(t)}
\end{equation*}
which is bounded by a $M^2/2$ according to the assumption on $\varphi$. With
\begin{equation*}
\partial_t \varrho(t)= a(t) P_\varphi a(t)+ \partial_t a(t) (\vs+t P_\varphi) a(t)+ a(t) (\vs+t P_\varphi) \partial_t  a(t), 
\end{equation*}
it follows by inspection that $\varrho(t) \in C^1(\Rm^+,\calJ_1)$ since $\vs$ and $P_\varphi$ are trace class and $a$ and $\partial_t a$ are bounded and continuous w.r.t $t$.

\subsection{Proof of Lemma \ref{convreg}} \label{proofconvreg}

The proof will use the following two ingredients: the first one is the logarithmic Sobolev \fref{logsobo}, and the second is the following Lemma proved in \cite[Lemma 3.1]{MP-JSP}, providing us with compactness results for sequences of density operators bounded in $\calE$.
\begin{lemma} \label{strongconv}
Let $(\varrho_k)_{k\in \NN}$ be  a bounded sequence of $\calE_+$. Then, up to an extraction of a subsequence, there exists $\varrho\in \calE_+$ such that
\be
\label{first}
\varrho_k\to\varrho\mbox{ in }\calJ_1\quad \mbox{as } k\to +\infty,
\ee
and
\be
\label{second}
\Tr(\sqrt{H_0}\varrho\sqrt{H_0})\leq \liminf_{k\to +\infty} \Tr(\sqrt{H_0}\varrho_k\sqrt{H_0}).
\ee
Furthermore, if one has
$$\Tr(\sqrt{H_0}\varrho\sqrt{H_0})= \lim_{k\to +\infty} \Tr(\sqrt{H_0}\varrho_k\sqrt{H_0}),$$
then one can conclude in addition that
\be
\label{third}
\sqrt{H_0}\sqrt{\varrho_k}\to \sqrt{H_0}\sqrt{\varrho_k}\mbox{ in }\calJ_2\quad \mbox{and}\quad\sqrt{H_0}\varrho_k\sqrt{H_0} \to \sqrt{H_0}\varrho\sqrt{H_0} \mbox{ in }\calJ_1 \quad \mbox{as } k\to +\infty.
\ee
\end{lemma}

\noindent\textbf{Step 1:} We start with $(i)$. Let $\sigma=\ket{\sqrt{n}} \bra{\sqrt{n}}$. Since $\sigma \in \calA$, we have
$$
F_\eta(\varrho_\eta) \leq F_\eta(\sigma).
$$
Besides, \fref{regbet} yields $F(\varrho_\eta)\leq F_\eta(\varrho_\eta)$. It follows from \fref{logsobo} and $\Tr(\varrho_\eta)=\|n\|_{L^1}=1$ that
\begin{equation*}
\int_{\Rm^d}n(x) \log(n(x))\,  dx-\frac{d}{2}\log\left(\frac{e E(\varrho_\eta)}{2\pi d}\right) - 1 \leq S(\varrho_{\eta}),
\end{equation*}
This gives the estimate
$$
\int_{\Rm^d}n(x) \log(n(x))\,  dx-\frac{d}{2}\log\left(\frac{e E(\varrho_\eta)}{2\pi d}\right) -1 +E(\varrho_\eta) \leq F(\varrho_\eta)\leq F_\eta(\sigma)= \|\nabla \sqrt{n}\|^2+\beta_\eta(1),
$$
which shows that 
\be \label{boundE}
E(\varrho_\eta) \leq M,
\ee
for some $M>0$ independent of $\eta$. Together with $\Tr(\varrho_\eta)=1$, we can apply Lemma \ref{strongconv} and find a subsequence $\varrho_k:=\varrho_{\eta_{\ell_k}}$ (recall that $\eta=\eta_\ell=1/\ell$) and a $\varrho$ satisfying the convergence results of Lemma \ref{strongconv}.

\noindent\textbf{Step 2:} We identify now $\varrho$ with $\vs$. For this, we remark first that $\vs \in \calA$, and therefore,
\be \label{ineqF}
F_\eta(\varrho_\eta) \leq F_\eta(\vs).
\ee
Furthermore, it is proven  in \cite[Step 6 in Section 3]{MP-KRM}, that $S(\varrho_k) \to S(\varrho)$ as $k \to \infty$, and the proof can be directly adapted to yield that
\be \label{geneconv}
\Tr\big(|\beta_{\eta_{\ell_k}}(\varrho_\ell)|\big) \underset{k\to\infty}\to \Tr\big(|\beta(\varrho)|\big) \qquad \textrm{and} \qquad \Tr\big(|\beta_{\eta_{\ell_k}}(\vs)|\big) \underset{k\to\infty}\to \Tr\big(|\beta(\vs)|\big).
\ee
As a consequence, we obtain from \fref{ineqF}  and \fref{second} of Lemma \ref{strongconv} that
$$
F(\varrho) \leq F(\vs),
$$
which, by uniqueness of the minimizer, yields $\varrho=\vs$. This also implies that the entire sequence $\{\varrho_{\eta_\ell}\}_{\ell \in \NN^*}$, denoted with an abuse of notation by $\{\varrho_{\ell}\}_{\ell \in \NN^*}$, converges to $\varrho_\star$. Then, (i) and (ii) follow from Lemma \ref{strongconv} by replacing $\varrho$ by $\vs$. Furthermore, (iii) and the first result of (iv) follow from \cite[Lemma A.2]{MP-JSP}. 

\noindent\textbf{Step 3:} We address now the second result of (iv). From \fref{boundE}, we have 
$$
\rho_{p,\eta} \| \nabla \phi_{p,\eta}\|^2 \leq \Tr  \big(\sqrt{H_0} \varrho_\eta\sqrt{H_0}\big) \leq M,
$$
which, together with (iii) and the first result of (iv), shows that $\sqrt{\varrho_{p,\ell}}  \phi_{p,\ell}$ converges weakly in $H^1(\Rm^d)$ to $\sqrt{\varrho_{p}}  \phi_{p}$. In order to obtain strong convergence, we remark that, according to (ii), $\sqrt{H_0} \sqrt{\varrho_\ell}$ converges to $\sqrt{H_0} \sqrt{\varrho}$ strongly in $\calJ_2$, so that 
$$
\rho_{p,\ell} \| \nabla \phi_{p,\ell}\|^2=\Tr \big(\sqrt{H_0} \sqrt{\varrho_\ell} P_{\phi_{p,\ell}} \sqrt{\varrho_\ell}
\sqrt{H_0}\big) \to \Tr \big(\sqrt{H_0} \sqrt{\varrho_\star} P_{\phi_{p}} \sqrt{\varrho_\star} \sqrt{H_0} \big)=\rho_{p} \| \nabla \phi_{p}\|^2.$$
Above, we used the fact that $P_{\phi_{p,\ell}} \to P_{\phi_{p}}$ strongly in $\calL(L^2(\Rm^d))$ because of the first result of (iv). Together with the weak convergence of $\sqrt{\varrho_{p,\ell}}\,  \nabla \phi_{p,\ell}$, this proves the second result of (iv).

\noindent\textbf{Step 4:} Regarding (v), as stated in \fref{geneconv}, we already have the convergence of $|\beta_{\eta_\ell}(\varrho_\ell)|$ in $\calJ_1$. With (iii) and (iv), it is then not difficult to obtain weak convergence of $\beta_{\eta_\ell}(\varrho_\ell)$ to $\beta(\vs)$ in the sense of operators, which, according to \cite[Theorem 2.21 and addendum H]{Simon-trace}, yields the convergence in $\calJ_1$.

\noindent\textbf{Step 5:} Finally, for (vi), we write
$$
x \log(\eta+x)-x \log(x) =\beta_\eta(x)-\beta(x)-g_\eta(x)\quad\textrm{with}\quad g_\eta(x)=\eta (\log (\eta)-\log(\eta+x)).
$$
Since $|g_\eta(x)| \leq |x|$ for all $x \geq 0$, we have that $g_\eta(\varrho_\eta) \in \calJ_1$. Then
$$
\|g_{\eta_\ell}(\varrho_\ell)\|_{\calJ_1}=\sum_{p \in \NN} |g_{\eta_\ell}(\rho_{p,\ell})| \underset{\ell \to \infty}\to 0, 
$$
as an application of generalized dominated convergence for series together with (iii), $g_{\eta_\ell}(\rho_{p,\ell}) \to 0$ and $\rho_p>0$ for all $p \in \Nm$. Hence, we have
$$
\| \varrho_\ell \log (\eta_\ell+\varrho_\ell)-\vs \log \vs\|_{\calJ_1} \leq \|\beta_{\eta_{\ell}}(\varrho_\ell)-\beta(\vs)\|_{\calJ_1}+\|g_{\eta_{\ell}}(\varrho_\ell)\|_{\calJ_1},
$$
which converges to zero as $\ell \to \infty$ according to (v). This ends the proof.

\subsection{Proof of Lemma \ref{rhoA}} \label{proofrhoa}
We can see that $a_\eta$ is bounded a.e. since, for any $t\in [-t_0,t_0]$,
  \be \label{bounda2}
  a_\eta^2(t) \leq \frac{\sum_{j \neq p,q} \rho_{j,\eta} |\phi_{j,\eta}|^2 + \rho_{p,\eta} |\phi_{p,\eta}|^2+\rho_{q,\eta} |\phi_{q,\eta}|^2}{\sum_{j \neq p,q} \rho_{j,\eta} |\phi_{j,\eta}|^2 + (\rho_{p,\eta}-|t|) |\phi_{p,\eta}|^2+(\rho_{q,\eta}-|t|) |\phi_{q,\eta}|^2} \leq 2.
  \ee
  This shows in particular that $\varrho_\eta(t)$ is trace class for $t \in [-t_0,t_0]$, and that $n[\varrho_\eta(t)]=n$. Furthermore, it is positive so that $\varrho_\eta(t)$ is a density operator for $t \in [-t_0,t_0]$. It remains to show that $\varrho_\eta(t) \in \calE$. We now follow the arguments of Step 2 of the proof of Lemma \ref{derivE}. We obtain first that
\begin{equation*}
E(a_{\eta}(t)\varrho_{\eta} a_{\eta}(t)) = \int_{\Rm^d} \left( n |\nabla a_{\eta}(t)|^2+ a_{\eta}(t )\nabla a_{\eta}(t) \cdot \nabla n+ a_{\eta}(t)^2 k[\varrho_{\eta}]\right)dx.
\end{equation*}
Furthermore, we have
\bee
E(a_{\eta}(t)P_{\eta} a_{\eta}(t)) &=& \int_{\Rm^d} 2\Re\left( \nabla(a_\eta(t)\phi_{p,\eta})\cdot \nabla(a_\eta(t)\phi_{q,\eta}^*)\right)dx\\
&=&  \int_{\Rm^d} \left(2\Re\left( \phi_{p,\eta}\phi_{q,\eta}^*\right)|\nabla a_\eta(t)|^2 + 2a_{\eta}(t )^2\Re\left( \nabla \phi_{p,\eta}\cdot \nabla\phi_{q,\eta}^*\right) \right)dx\\
&&+\int_{\Rm^d} 2a_{\eta}(t)\nabla a_{\eta}(t)\cdot \Re\left( \phi_{p,\eta} \nabla \phi_{q,\eta}^*+ \phi_{q,\eta}^*\nabla \phi_{p,\eta} \right) dx,
\eee
leading to
\begin{align*}
E(\varrho_{\eta}(t)) =\;&  \int_{\Rm^d}  (n_{\eta}(t) \nabla a_{\eta}(t)+ a_{\eta}(t ) \nabla n_{\eta}(t))\cdot\nabla a_{\eta}(t) dx
\\ &+ \int_{\Rm^d} \left( a_{\eta}(t)^2( k[\varrho_{\eta}]+ 2 t \Re\left( \nabla \phi_{p,\eta}\cdot \nabla\phi_{q,\eta}^*\right)\right)dx.
\end{align*}
By using the fact that
\begin{equation*}
\nabla a_{\eta}(t)= \frac{\nabla \sqrt{n}}{\sqrt{n_{\eta}(t)}}-\frac{a_{\eta}(t)\nabla\sqrt{n_{\eta}(t)} }{\sqrt{n_{\eta}(t)}},
\end{equation*}
we deduce the expression
\begin{equation}\label{eq:Evarrhoeta}
  E(\varrho_{\eta}(t)) =  \int_{\Rm^d} \left( |\nabla \sqrt{n}|^2 - a_{\eta}(t)^2\left|\nabla\sqrt{n_{\eta}(t)}\right|^2 + a_{\eta}(t)^2 \left(k[\varrho_{\eta}]+2 t \Re\left( \nabla \phi_{p,\eta}\cdot \nabla\phi_{q,\eta}^*\right)\right) \right)dx.
\end{equation}
Finally, since the following estimate holds for any $t\in[-t_0,t_0]$
\begin{equation*}
|n_\eta(t)| = |n + 2t\Re(\phi_{p,\eta}\phi_{q,\eta}^*)|\geq n/2,
\end{equation*}
we deduce that
\begin{align*}
\left|\nabla\sqrt{n_{\eta}(t)}\right| &\leq \frac{|\nabla n| + \rho_{p,\eta}|\phi_{p,\eta} \nabla \phi_{q,\eta} | + \rho_{q,\eta}|\phi_{q,\eta} \nabla \phi_{p,\eta} |}{n}
\\ &\leq \frac12|\nabla\sqrt{n}| + |\nabla \phi_{p,\eta} | + |\nabla \phi_{q,\eta} |,
\end{align*}
which enables us to bound each term of \eqref{eq:Evarrhoeta} since $\phi_{p,\eta}$ and $\phi_{q,\eta}$ belong to $H^1(\Rm^d)$ according to \fref{KE}. This shows that $\varrho_{\eta}(t)\in\mathcal{E}$ for $t\in[-t_0,t_0]$ and concludes the proof.

    \bibliographystyle{plain}
\bibliography{bibliography.bib}

\begin{thebibliography}{10}

\bibitem{QHD-CMS}
P.~Degond, S.~Gallego, and F.~M{\'e}hats.
\newblock On quantum hydrodynamic and quantum energy transport models.
\newblock {\em Commun. Math. Sci.}, 5(4):887--908, 2007.

\bibitem{QET}
P.~Degond, F.~M{\'e}hats, and C.~Ringhofer.
\newblock Quantum energy-transport and drift-diffusion models.
\newblock {\em J. Stat. Phys.}, 118(3-4):625--667, 2005.

\bibitem{QHD-review}
P.~Degond, F.~M{\'e}hats, and C.~Ringhofer.
\newblock Quantum hydrodynamic models derived from the entropy principle.
\newblock In {\em Nonlinear partial differential equations and related
  analysis}, volume 371 of {\em Contemp. Math.}, pages 107--131. Amer. Math.
  Soc., Providence, RI, 2005.

\bibitem{DR}
P.~Degond and C.~Ringhofer.
\newblock Quantum moment hydrodynamics and the entropy principle.
\newblock {\em J. Statist. Phys.}, 112(3-4):587--628, 2003.

\bibitem{Dolbeault-Loss}
J.~Dolbeault, P.~Felmer, M.~Loss, and E.~Paturel.
\newblock Lieb-{T}hirring type inequalities and {G}agliardo-{N}irenberg
  inequalities for systems.
\newblock {\em J. Funct. Anal.}, 238(1):193--220, 2006.

\bibitem{DP-JMPA}
R.~Duboscq and O.~Pinaud.
\newblock On the minimization of quantum entropies under local constraint.
\newblock {\em Submitted}.

\bibitem{jungelbook}
A.~J{\"u}ngel.
\newblock {\em Quasi-hydrodynamic semiconductor equations}.
\newblock Progress in Nonlinear Differential Equations and their Applications,
  41. Birkh\"auser Verlag, Basel, 2001.

\bibitem{jungel-matthes}
A.~J{\"u}ngel and D.~Matthes.
\newblock A derivation of the isothermal quantum hydrodynamic equations using
  entropy minimization.
\newblock {\em ZAMM Z. Angew. Math. Mech.}, 85(11):806--814, 2005.

\bibitem{jungel-matthes-milisic}
A.~J{\"u}ngel, D.~Matthes, and J.~P. Mili{\v{s}}i{\'c}.
\newblock Derivation of new quantum hydrodynamic equations using entropy
  minimization.
\newblock {\em SIAM J. Appl. Math.}, 67(1):46--68, 2006.

\bibitem{levermore}
C.~D. Levermore.
\newblock Moment closure hierarchies for kinetic theories.
\newblock {\em J. Statist. Phys.}, 83(5-6):1021--1065, 1996.

\bibitem{MP-JSP}
F.~M{\'e}hats and O.~Pinaud.
\newblock An inverse problem in quantum statistical physics.
\newblock {\em J. Stat. Phys.}, 140(3):565--602, 2010.

\bibitem{MP-KRM}
F.~M{\'e}hats and O.~Pinaud.
\newblock A problem of moment realizability in quantum statistical physics.
\newblock {\em Kinet. Relat. Models}, 4(4):1143--1158, 2011.

\bibitem{RS-80-I}
M.~Reed and B.~Simon.
\newblock {\em Methods of modern mathematical physics. I. Functional analysis}.
\newblock Academic Press, Inc., New York, second edition, 1980.

\bibitem{RS-80-2}
M.~Reed and B.~Simon.
\newblock {\em Methods of modern mathematical physics. II. Fourier Analysis,
  Self-Adjointness}.
\newblock Academic Press, Inc., New York, second edition, 1980.

\bibitem{Simon-trace}
Barry Simon.
\newblock {\em Trace ideals and their applications}, volume 120 of {\em
  Mathematical Surveys and Monographs}.
\newblock American Mathematical Society, Providence, RI, second edition, 2005.

\end{thebibliography}

\end{document}